\documentclass{article}
\usepackage{amssymb,amsmath,bm}
\usepackage{amsthm,amsfonts}
\usepackage{color}

\usepackage{latexsym,amscd}
\usepackage{graphicx,color,subfigure}
\usepackage{comment}
\usepackage{float}
\usepackage{mathrsfs}
\usepackage{rotating}
\usepackage{multirow}

\newcommand\bone{\mathbf{1}}

\newcommand{\R}{\mathbb{R}}
\newcommand{\RR}{\mathbb{R}}

\newtheorem{theorem}{Theorem}[section]
\newtheorem{corollary}[theorem]{Corollary}

\newtheorem{remark}[theorem]{Remark}

\definecolor{brilliantrose}{rgb}{1.0, 0.33, 0.64}
\definecolor{amber}{rgb}{1.0, 0.75, 0.0}
\definecolor{amethyst}{rgb}{0.6, 0.4, 0.8}
\definecolor{carrotorange}{rgb}{0.93, 0.57, 0.13}

\title{A Theory for Backtrack-Downweighted  Walks}
\author{Francesca Arrigo, Desmond J. Higham and Vanni Noferini}
\date{}

\author{
Francesca Arrigo\thanks{Department of Mathematics and Statistics, University of Strathclyde, Glasgow, UK, G1 1XH}
\and
Desmond J. Higham\thanks{School of Mathematics, University of Edinburgh,
James Clerk Maxwell Building,
Edinburgh, UK, EH9 3FD} \and 
Vanni Noferini\thanks{Aalto University, Department of Mathematics and Systems Analysis, P.O. Box 11100, FI-00076, Aalto, Finland}
}

\begin{document}

\maketitle

 \begin{abstract}
We develop a complete theory 
for the combinatorics of 
walk-counting on a directed graph 
in the case where each
backtracking step is downweighted
by a given factor.
By deriving expressions for the 
associated generating functions, we  
also obtain 
linear systems for computing 
centrality measures in this setting.
In particular,   
 we show that
 backtrack-downweighted 
  Katz-style network 
  centrality can be computed at the same cost 
   as standard Katz.
   Studying the limit of this 
    centrality measure at its radius of convergence 
    also leads to a new expression for 
     backtrack-downweighted  eigenvector 
     centrality that generalizes
     previous work to the case where directed edges are present.
        The new theory allows us to 
       combine advantages 
       of standard and nonbacktracking 
       cases, avoiding localization while accounting for tree-like 
        structures. 
     We illustrate the behaviour of the
       backtrack-downweighted centrality measure on both synthetic
      and real networks. 
 \end{abstract}

\section{Motivation}
\label{sec:mot}

\subsection{Nonbacktracking Walks}
Many concepts in network science are built on the notion of walks.
We may study the transient or long term behaviour 
of random walks of a certain form.
Or we may consider the combinatorics of all distinct walks of a given type.
Such ideas, which lie at intersection between 
graph theory, 
combinatorics and 
applied linear algebra, 
 form the basis of effective algorithms for 
summarizing network properties
\cite{bookKE,Newmanbook}.

Our focus here is on the definition and analysis of a 
new type of combinatoric walk-count which interpolates between 
the classical and nonbacktracking cases.
In the classical setting, a walker may continue by traversing any edge pointing out of the current node.
In the nonbacktracking setting, the walker 
must never continue along the reverse 
 of the edge on which they arrived. 
 Intuitively, eliminating backtracking forces the walker to 
 explore the network more widely.
 More concretely, it has been shown to offer benefits
 in 
 centrality measurement
\cite{AGHN17a,AGHN17b,CFGGP19,GHN18,MZN14,PC16,TCTE20}, 
community detection \cite{Kawa16,KMMNSZZ13,New13,PGNRU20,SH15}, 
network comparison and alignment 
\cite{MOW17,TSE19} 
and in the study of related issues concerning 
optimal percolation \cite{MM15,MMBMM16} and the spread of
epidemics \cite{MR20}.

Nonbacktracking also plays an important role in a number of 
seemingly unrelated scientific fields,
including 
spectral graph theory \cite{AFH15,H07}, 
 number theory \cite{Te13},
discrete mathematics \cite{BL70,HST,ST96},
quantum chaos \cite{Sm07},
random matrix theory \cite{So07}, and 
computer science \cite{SKZ14,WF09}.
Hence, our results also have potential for impact outside network science.

\subsection{Downweighting rather than Eliminating} \label{sec:down}
Our work can be motivated by two issues
\begin{description}
\item[a.] Treating all walks equally leads to centrality measures that  
suffer from localization---placing most emphasis on a small subset of nodes and 
struggling to distinguish between the remainder. 
This issue can be associated with an accumulation of 
backtracking walks 
\cite{AGHN17a,MZN14}. 
\item[b.] Completely eliminating backtracking walks, however, 
may overlook 
some features, notably the existence of trees
\cite{New13}. 
 \end{description}
 For this reason, 
we will consider a more general regime where  
any backtracking step during a walk is downweighted
by some factor, $0 \le \theta \le  1$. So the extremes of
$\theta  = 1$
and 
$\theta  = 0$
correspond to 
standard and nonbacktracking walk counts, respectively.

We are concerned with the combinatorics
of such  backtrack-downweighted walks---we seek a formula for the number of distinct
walks of each length in a given graph. 
We then study the associated generating functions in order to
produce  Katz-style network centrality measures.
Moreover, by considering how the resolvent-based generating function behaves at its radius of convergence, we also arrive at a corresponding eigenvector centrality.

We note that in \cite{CFGGP19} the related concept of 
alpha-nonbacktracking\footnote{We prefer to use a different symbol,
$\theta$, for the backtrack-downweighting parameter, since in our context 
$\alpha$ is traditionally used for attenuation.}
centrality was introduced. That work  
focused on the eigenvector setting, adapting the Hashimoto matrix construction,
and applied only to undirected networks.
Our work allows for directed networks and is built on a combinatoric
walk-counting approach that generalizes Katz centrality.
Further, by working at the node level rather than the edge level, we are able to 
derive more computationally efficient measures, based on linear systems with the same dimension and sparsity as the original network.

We finish this section by describing how the manuscript is organized, and 
pointing out how previous results are generalized.
In section~\ref{sec:bg} we introduce the required
background concepts.
Section~\ref{sec:btdw}
defines 
backtrack-downweighted walks
and 
in Theorem~\ref{thm:qres} we 
derive a general  four-term recurrence
that allows us to count them.
This result
extends the fully nonbacktracking version from
\cite{BL70}; see also \cite{TP09}.
Section~\ref{sec:res} concerns the standard generating function.
Corollary~\ref{cor:qKatz} gives a linear system for the 
associated Katz-style centrality measure,
generalizing \cite[equation (3.3)]{AGHN17a}.
In section~\ref{sec:squid} we give results on a small graph 
that illustrate the  two motivational issues 
\textbf{a} and \textbf{b} above.
(This graph has an interesting spectral property that
is described in Remark~\ref{rem:squid}.)
We 
also give comparative results 
on the star 
graph and 
on 
regular graphs. 
In section~\ref{sec:spec} we consider the limit
as the Katz parameter
approaches
its upper value, and thereby, in Theorem~\ref{thm:spec}, derive 
an eigenvector centrality measure. This result extends
 the measure in 
 \cite{CFGGP19} to the case of directed graphs.
 Section~\ref{sec:expo}
 shows how the recurrence from
 Theorem~\ref{thm:qres} 
 can be used to compute generating functions based on general power series.
 Here we find it necessary to work with  block matrices of three times the dimension of the original adjacency matrix; see Theorem~\ref{thm:genfun}, which extends
 \cite[Theorem~5.2]{AGHN17b}.
 In section~\ref{sec:compex} we give results on data from the London Underground
train network, and we argue that the backtrack-downweighting parameter 
provides a useful means to mitigate localization while maintaining 
 correlation with passenger usage. 
We finish with a brief discussion in section~\ref{sec:disc}.

\section{Background and Notation}
\label{sec:bg}

We consider an unweighted, directed network 
with $n$ nodes. We let $A \in \R^{n \times n}$ 
denote the adjacency matrix, so $a_{ij} = 1$ if there is an edge from $i$ to $j$ and
$a_{ij} = 0$ otherwise. There are no self-loops, so $a_{ii} = 0$.
A \emph{walk of length $k$} from node $i$ to node $j$
is a sequence of nodes $i=i_0, i_1, i_2, \ldots, i_k = j$ such that each edge
from $i_s$ to $i_{s+1}$ exists. Note that the nodes in the sequence 
are not required to be distinct; the walk may revisit nodes and edges.
It follows directly from the definition of matrix multiplication that 
$(A^{k})_{ij}$ counts the number of distinct walks of length $k$ 
from $i$ to $j$ \cite{bookKE}.

Katz \cite{Katz53} used this walk-counting expression as the basis for a centrality measure.
Here, we compute a value $x_i > 0$ that quantifies the importance of node $i$, with a larger value indicating greater importance.
Katz centrality uses 
\begin{equation}
x_i = 1 + \sum_{k=1}^{\infty}  \sum_{j=1}^{n}\alpha^k (A^k)_{ij}.
\label{eq:katzi}
\end{equation}
Here, (up to a convenient constant unit shift) 
the centrality of node $i$ is given by the 
total number of walks from node $i$ to every node, with a walk of length 
$k$ weighted by a factor  $\alpha^k$, where $\alpha>0$ is a real parameter.
  This series converges for $\alpha < 1/\rho(A)$, where $\rho(\cdot)$ denotes the
   spectral radius, and we may use the matrix-vector notation 
   \begin{equation}
    (I - \alpha A) \mathbf{x} = \mathbf{1},
    \label{eq:katz}
    \end{equation}
    where $\mathbf{1}  \in \R^{n}$  has all elements equal to one.

A \emph{nonbacktracking walk of length $k$}
from node $i$ to node $j$
is a sequence of nodes $i=i_0, i_1, i_2, \ldots, i_k=j$ such that each edge
from $i_s$ to $i_{s+1}$ exists and
we never have $i_{s} = i_{s+2}$. In words, after leaving a node we must not 
return to it immediately.
Now, let $p_k(A) \in \R^{n \times n}$ be such that 
$(p_k(A))_{ij}$ records the number of distinct nonbacktracking
 walks of length $k$ 
from $i$ to $j$.
It is straightforward to show that 
\[
p_1(A) = A, \quad \text{and} \quad p_2(A) = A^2 - D,
\]
where $D \in \R^{n \times n}$ is the diagonal matrix whose entries are $d_{ii} = (A^2)_{ii}$.
Setting 
$p_0(A) = I$
for convenience,
it was shown in
\cite{BL70}
that
the matrices 
$p_k(A)$ satisfy the following four-term recurrence
\begin{equation}
p_{k+1}(A) = p_k(A)A + p_{k-1}(A)(I-D) - p_{k-2}(A)(A-S),
\quad 
\text{for~} k = 2,3,\ldots, 
\label{eq:pkrec}
\end{equation}
where $S \in \R^{n \times n}$ is such that 
$s_{ij} = a_{ij} a_{ji}$. 
See \cite{TP09} for an alternative, linear algebraic, proof.

 We note that in the extreme case of a directed network
 for which no edges are reciprocated,
 that is $a_{ij} = 1 \Rightarrow a_{ji} = 0$, there is no opportunity for any walk to backtrack; all walks are nonbacktracking.
 In this case, we have 
 $D = 0$ and $S = 0$, and 
$p_k(A)$ recovers  
the classical walk count $A^{k}$.

 Motivated by (\ref{eq:katzi}),
the nonbacktracking Katz analogue  
  \begin{equation}
  x_i = 1 + \sum_{k=1}^{\infty}  \sum_{j=1}^{n}\alpha^k (p_k(A))_{ij} 
\label{eq:katzinbtw}
    \end{equation}
    was introduced in \cite{AGHN17a}.
    Here, centrality is computed via weighted combinations of
    nonbacktracking (rather than classical) walks.  
    By first deriving an expression for the generating function
    $\sum_{k=0}^{\infty} \alpha^k p_k(A)$,
     it was shown that
     (\ref{eq:katzinbtw}) solves the linear system
      \begin{equation}
    \left(I - \alpha A  + \alpha^2 (D-I) - \alpha^3 (A-S) \right) \mathbf{x} = 
  \left( 1 - \alpha^2 \right)    
    \mathbf{1}. 
    \label{eq:katznbtw}
    \end{equation}
    In general, the radius of convergence for the 
series in (\ref{eq:katzinbtw}), and hence the range of valid 
$\alpha$ values in (\ref{eq:katzinbtw}), is governed by the spectrum of a 
three-by-three block matrix; see \cite[Theorem~5.1]{AGHN17a} and 
section~\ref{sec:spec}.

   The
       nonbacktracking version of Katz centrality, (\ref{eq:katzinbtw}), 
       was first defined and analyzed in 
           \cite{GHN18} for     
       undirected networks. In this undirected case,   
      as $\alpha$ approaches its upper limit the ranking induced by the 
      centrality measure in (\ref{eq:katznbtw}) generically tends 
    to that induced by the nonbacktracking       
      eigenvector centrality measure introduced in 
      \cite{MZN14}; see \cite[Theorem~10.2]{GHN18}.
      Taking the corresponding limit in 
      (\ref{eq:katznbtw}) 
      defines a computable 
      nonbacktracking       
      eigenvector centrality
       for the more general case of a directed network, 
       \cite[Theorem~6.1]{AGHN17a}.

\section{Backtrack-Downweighted Walk Counts}\label{sec:btdw}

We now consider an intermediate regime where backtracking is not completely 
eliminated, but rather the count for each walk is downweighted by 
$\theta^m$ where $0 \le \theta \le 1$ is a parameter and 
$m$ is the number of backtracking steps incurred during the walk.
Hence, $\theta =1$ corresponds to the classical walk count
 $A^k$ and 
 $\theta =0$ corresponds to the 
 nonbacktracking walk count from 
 $p_{k}(A)$ 
in (\ref{eq:pkrec}).
We will let $q_{k}(A)$ denote the resulting 
\emph{backtrack-downweighted walk}
(BTDW) count matrix, where, for brevity, 
the dependence of 
$q_k(A)$ on $\theta$ is not explicitly indicated. 
More precisely, 
$(q_k(A))_{ij}$ counts the number of distinct walks
of length $k$ 
from node $i$ to node $j$ with the following
proviso:
for each walk,  $i=i_0, i_1, i_2, \ldots, i_k=j$,
every occurrence of a backtracking step
($i_{s} = i_{s+2}$)
incurs a downweighting by a factor $\theta$.

To illustrate this idea, consider the directed graph in 
Figure~\ref{fig:graph5}.
Looking at walks of length four from node $1$ to node $5$, we have
\begin{itemize}
\item a walk $1 \to 2 \to 3 \to 4 \to 5$ with no backtracking,
\item a walk $1 \to 2 \to 3 \to 2 \to 5$ with one instance of backtracking,
\item a walk $1 \to 2 \to 5 \to 2 \to 5$ with two instances of backtracking.
\end{itemize}
Hence, $(q_{4}(A))_{15} = 1 + \theta + \theta^2$.
Continuing with these arguments, we find that 
\[
q_{4}(A) =
 \left[
  \begin{array}{ccccc}
     0 & 0 & \theta + \theta^2 & 0 & 1 + \theta + \theta^2 \\
      0 & 1 + 2 \theta^2 + 2 \theta^3 & 0 & \theta + \theta^2 & 0 \\
       0 & 0 & 1+ \theta + \theta^3 & 0 & 2 \theta + 2 \theta^2 \\
        0 & \theta + \theta^2 & 0 & 1 & 0 \\
       0 & 0 & 2 \theta^2 & 0 & 1 + \theta + \theta^3
       \end{array}
       \right].
       \]

\begin{figure}
\begin{center}
\includegraphics[width=0.25\textwidth]{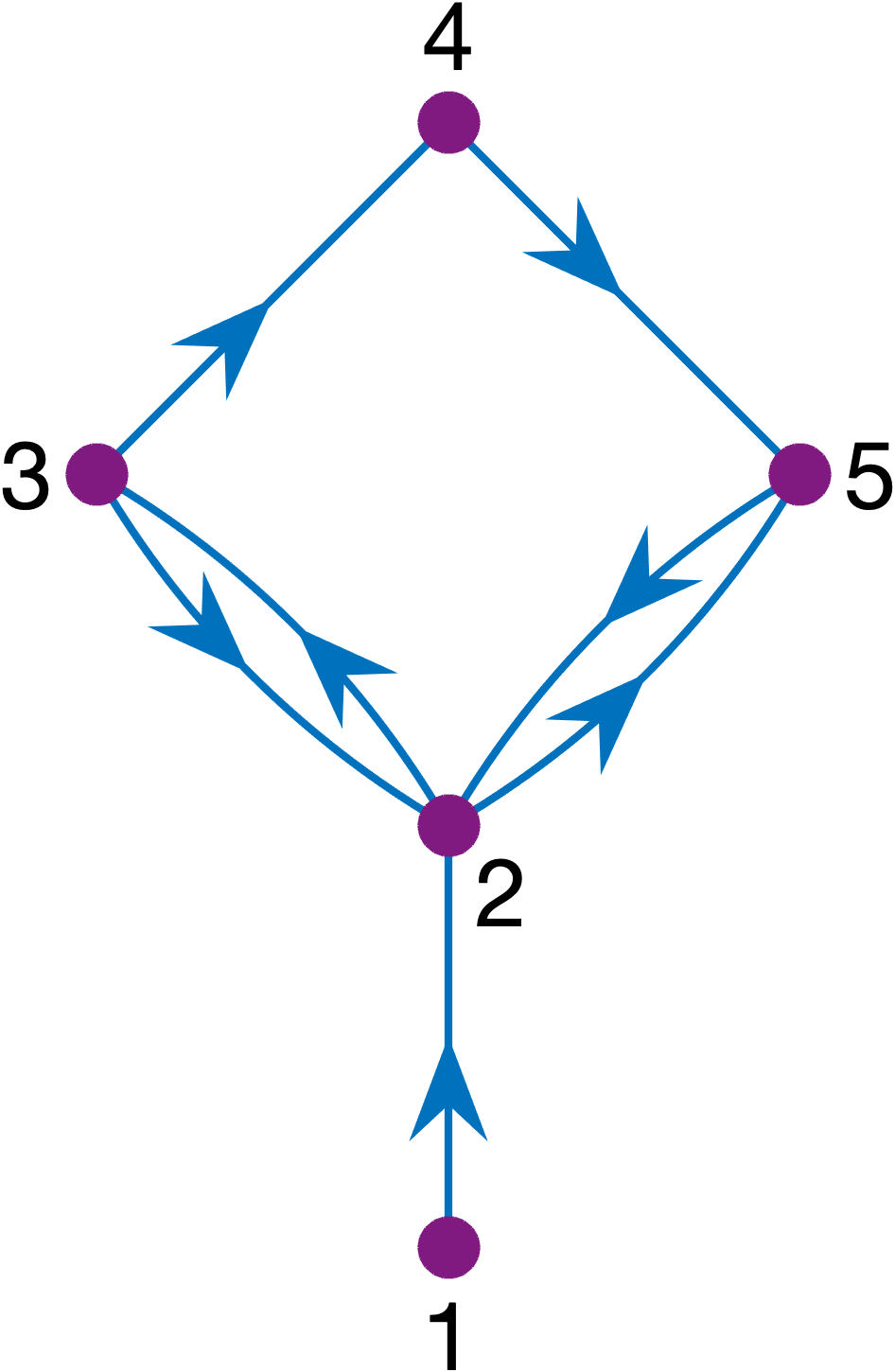}
\end{center}
\caption{Illustrative directed graph.}
\label{fig:graph5}
\end{figure}

The following theorem generalizes the recurrence 
(\ref{eq:pkrec}). 
For the statement of this theorem, and 
later results, we 
find it convenient to
let $\mu = 1 - \theta$. 

\begin{theorem} \label{thm:qres}
Letting $q_0(A) = I$, 
we have 
\[
q_1(A) = A, \qquad q_2(A) = A^2 - \mu D,
\]
and for $k\geq 2$ 
\begin{equation}
q_{k+1}(A) = q_k(A) A + \mu \, q_{k-1}(A) \left( \mu I - D\right)  
    -  \mu^2 \, q_{k-2}(A) \left(A - S \right),
    \label{eq:qrec}
\end{equation}
where $\mu = (1-\theta)$.
\end{theorem}

\begin{proof}

The identity $q_1(A) = A$ follows immediately since no walk of length one can backtrack.
Backtracking walks of length two are precisely closed walks of length two.
Hence, the off-diagonal elements of
$q_2(A)$ match those of 
 $A^2$, and 
the diagonal  
elements of
$q_2(A)$ correspond to those of 
$A^2$ scaled by $\theta$. This gives 
 $q_2(A) = A^2 - D + \theta D$.

We proceed by induction.
Assume that 
$q_s(A)$ correctly counts  BTDWs of length $s$ for $s \le k$.
We start with the expression
\begin{equation}
q_k(A) A.
\label{eq:try1}
\end{equation}
Postmultiplying by $A$ in this way corresponds to adding an edge to the end of the walk: 
the $(i,j)$ entry of $q_k(A) A$ deals with walks from $i$ to $j$ of length $k+1$ where any backtracking arising along the first $k$ edges has been correctly downweighted. 
We may associate 
$(q_k(A) A)_{ij}$ 
with the schematic 
\begin{equation}
      \underbrace{ i \,  \star \, \star \, \cdots  \, \star \, \star}_{k\checkmark} \, j.
      \label{eq:s1}
\end{equation}
Here $i$ and $j$ are the first and last nodes in the walk of length $k+1$.
The star symbols denote arbitrary nodes.
The presence of $k\checkmark$ indicates that, in every walk
under consideration, the first $k$ edges along that walk have been 
correctly downweighted (since they are accounted for by
  $q_k(A)$).
  However, any backtracking caused by the final edge has not been 
  correctly 
  downweighted. 
  The walks whose weights we must adjust have the form
  \begin{equation}
    \underbrace{ i \,  \star \,   \star \,  \cdots   \star \, \, \star \, j \, \ell  }_{k\checkmark} \, j. 
      \label{eq:s2}
\end{equation}
 To deal with such walks we note that the term
$(q_{k-1}(A) D)_{ij}$ is associated with walks of the form
  \begin{equation}
    \underbrace{ i \,  \star \,   \star \, \cdots \,  \star \,  \star \, j  }_{(k-1)\checkmark} \, \ell \, j. 
\label{eq:s3}
\end{equation}
These walks appeared in our original expression, $q_k(A) A$ in (\ref{eq:try1}), without 
any downweighting of the final backtracking step. 
We may therefore remove the contribution from all such walks to 
this expression, to
give 
 $q_k(A) A - q_{k-1}(A) D$, and then add this contribution  back in with the 
 required extra downweighting factor $\theta$, leading to the expression 
 \begin{equation}
 q_k(A) A + (\theta - 1)q_{k-1}(A) D.
 \label{eq:try2}
\end{equation}

 However,  the walks represented by (\ref{eq:s2}) and 
 (\ref{eq:s3}) are not the same, because we have not yet properly dealt with 
 walks of the form 
   \begin{equation}
    \underbrace{ i \,  \star \,   \star \,  \cdots  \star \,  \, \ell \, j \, \ell  }_{k \checkmark} \, j.
\label{eq:s4}
\end{equation}
To proceed, we note that 
$(q_{k-1}(A))_{i j}$ contains the correct BTDW count for walks of length $k-1$
from $i$ to $j$; that is,
  \begin{equation}
    \underbrace{ i \,  \star \,   \star \,  \cdots  \star \,   j  }_{(k -1)\checkmark}.
\label{eq:s5}
\end{equation}
To see how these may extend to walks of the form 
  \[
    \underbrace{ i \,  \star \,   \star \,  \cdots  \star \,  \ell \,  j  }_{(k -1)\checkmark} \, \ell \, j 
\]
we consider two cases
\begin{description}
\item[1.]
if the reciprocated edge from $j$ to $\ell$ exists, then there is exactly one such walk,
\item[2.]
if the reciprocated edge from $j$ to $\ell$ does not exist, then there is no such walk.
\end{description}
The quantity 
$(q_{k-1}(A) - q_{k-2}(A)(A-S))_{ij}$ therefore accounts for these walks, but with a scaling that does not allow for the final two backtracking steps.

Hence the correctly weighted contribution 
to $q_{k+1}(A)$ from walks of the form  
(\ref{eq:s4}) 
is $\theta^2 [  q_{k-1}(A) -    q_{k-2}(A) (A-S)]_{\ell j}$. The factor 
$\theta^2$
is required because of the two
final backtracking steps, which are not downweighted in $q_{k-1}(A) -   q_{k-2}(A) (A-S)$. 
In order to make (\ref{eq:try2}) correct we therefore need to
\begin{itemize}
 \item  subtract the amount 
 $(\theta - 1 ) [ q_{k-1}(A) -    q_{k-2}(A) (A-S)]$ 
 in order to compensate for the fact that these walks were incorrectly scaled by a factor
 $1-\theta$, rather than $\theta^2$, in $q_{k-1}(A) D$,
 \item subtract the amount $\theta [ q_{k-1}(A) -   q_{k-2}(A) (A-S)]$ in order to compensate for the
 fact that these walks were incorrectly scaled by a factor $\theta$ rather than 
 $\theta^2$ in $q_k(A)A$,
 \item (having now removed the contribution from these walks) 
  add in $\theta^2  [q_{k-1}(A) -   q_{k-2}(A) (A-S)]$ so that the two final backtracking steps 
  are accounted for properly.
  \end{itemize}
 This leads us to the relation 
\[
 q_{k+1}(A) =   q_k(A) A + (\theta - 1)q_{k-1}(A) D
 + (1 - \theta)^2 
 [q_{k-1}(A) -    q_{k-2}(A) (A-S)],
 \]
giving the required result.
\end{proof}

The following corollary gives an alternative version of the recurrence.

\begin{corollary}
For $k\geq 2$,
the BTDW count matrices $q_{k}(A)$ also satisfy the recurrence 
\begin{equation}
q_{k+1}(A) =  A q_k(A) + \mu \, \left( \mu  I - D\right)   q_{k-1}(A)  
    -  \mu^2 \, \left(A - S \right) q_{k-2}(A).
    \label{eq:qrecb}
\end{equation}
\end{corollary}

\begin{proof}
The recurrence (\ref{eq:qrecb}) can be established by adapting the proof of 
Theorem~\ref{thm:qres}.
The main difference is to add the $(k+1)$st edge at the beginning of the walk,
rather than the end;  
so, instead of (\ref{eq:try1}), we begin with $A  q_k(A)$.  

However, a higher-level argument can also be used.
Reversing the direction of every
edge in a graph, the BTDW count
from   
from $i$ to $j$ becomes the BTDW count from $j$ to $i$.
Hence, $q_{k}(A^T) = ( q_{k}(A) )^T$.
Writing the recurrence (\ref{eq:qrec}) for $A^T$ and taking the transpose, we then arrive at  (\ref{eq:qrecb}).
\end{proof}

 \begin{remark}
 It is straightforward to confirm that 
$\theta = 0$ in  (\ref{eq:qrec}), that is, elimination of all backtracking walks, 
leads to the relation (\ref{eq:pkrec}).
Also, $\theta = 1$, treating al walks equally, gives the classical count 
$A^{k}$ for walks of length $k$.
 \end{remark}

 \section{Resolvent}\label{sec:res}
 
 To study Katz-style centrality, we define the generating function
  \begin{equation}
\Psi(A) = \sum_{k=0}^{\infty}  \alpha^k q_{k}(A).
\label{eq:psi}
\end{equation}

\begin{theorem} \label{thm:genres}
If $\alpha \ge 0$ is within the radius of convergence of the 
generating function
$\Psi(A)$ in (\ref{eq:psi}) then
\[
\Psi(A) \left[1 - \alpha A - \mu \alpha^2 
   \left( \mu I - D \right)
    +
     \mu^2 \alpha^3 \left(A-S\right) \right]
      = \left(
           1 - \mu^2 \alpha^2 \right) I,
           \]
where we recall that $\mu = 1-\theta$.
\end{theorem}

\begin{proof}  
Multiplying by $\alpha^{k+1}$ in (\ref{eq:qrec}) and summing from $k=2$ to
$\infty$ gives
\begin{eqnarray*}
  \sum_{k=2}^{\infty}
     \alpha^{k+2} q_{k+1}(A)
       &=&
        \alpha 
      \sum_{k=2}^{\infty}
     \alpha^{k} q_{k}(A)  A
      + 
       \mu \alpha^2 
         \sum_{k=2}^{\infty}
     \alpha^{k-1} q_{k-1}(A)   \left( \mu I - D \right)
     \\
     && \mbox{} - \mu^2 
       \alpha^3
         \sum_{k=2}^{\infty}
     \alpha^{k-2} q_{k-2}(A)   \left(A-S\right).
  \end{eqnarray*}   
  Hence,
        \begin{eqnarray*}
        \Psi(A)
         - \alpha^2 q_{2}(A) - \alpha q_1(A) - q_0(A) 
          &=&
                \alpha \left( \Psi(A) - \alpha q_1(A) - q_0(A) \right)  A\\
                &&  \mbox{} +
                   \mu \alpha^2  \left( \Psi(A) - q_0 \right) \left( \mu I - D \right) \\
                      && \mbox{} -
                         \mu^2 \alpha^3   \Psi(A) \left(A-S\right).
                           \end{eqnarray*}  
                            Using the expressions for 
  $q_0(A)$, $q_1(A)$ and $q_2(A)$ in Theorem~\ref{thm:qres},   
 the result follows
                          after rearrangement and simplification. 
  \end{proof}

Following (\ref{eq:katzi})
and 
(\ref{eq:katzinbtw}), we define the BTDW Katz centrality
for node $i$ as
  \begin{equation}
  x_i = 1 + \sum_{k=1}^{\infty}  \sum_{j=1}^{n}\alpha^k (q_k(A))_{ij}.
\label{eq:katzibtdw}
    \end{equation}  
    We may then generalize (\ref{eq:katznbtw}) as follows. 
  
\begin{corollary} \label{cor:qKatz}  
The BTDW Katz centrality measure (\ref{eq:katzibtdw}) solves the linear system
  \begin{equation}
 \left[
    I - \alpha A - \mu \alpha^2   \left( \mu I - D \right) +
    \mu^2 \alpha^3 \left(A-S\right)  \right]
    \mathbf{x} 
     = 
       \left( 
          1 - \mu^2 \alpha^2 \right) 
          \mathbf{1}.
          \label{eq:katzbtdw}
    \end{equation}  
  \end{corollary}

\begin{proof}
From (\ref{eq:katzibtdw}) we have 
$\Psi(A)^{-1}\mathbf{x} = \mathbf{1}$.
Theorem~\ref{thm:genfun} gives the required expression for 
$\Psi(A)^{-1}$.
\end{proof}

We note that the  coefficient matrix in (\ref{eq:katzbtdw}) has the same sparsity 
as the coefficient matrix in standard Katz
(\ref{eq:katz}).
This shows that (a) downweighting of backtracking walks can be incorporated at
no extra computational cost, and (b) it is 
therefore feasible to apply the measure to large, sparse networks.
Also, by construction, the 
radius of convergence in (\ref{eq:psi}) for a general 
$ 0 < \theta < 1$ must be bounded above and below by the corresponding 
radius of convergence when $\theta = 1$ and $\theta = 0$, respectively.

  \section{Squid, Star and Regular Graphs}\label{sec:squid}
  We now analyze specific simple examples that shed light on 
  how the 
   BTDW Katz centrality measure 
    in (\ref{eq:katzbtdw}) can perform differently to the 
   extreme cases of standard and 
   fully nonbacktracking Katz.
   We first consider the undirected graph with $11$ nodes  
    shown in Figure~\ref{fig:squid}.
    Due to its shape, we will refer this as the squid graph.
    Here, node 1 has the highest degree, but it could be argued that nodes 6 
     and 8, of lower degree, possess better quality connections. 
      In particular, node 1 is connected to four leaves, and the 
      subgraph consisting of nodes 1,2,3,4 and 5
    represents a tree hanging off the remainder of the graph.  
     In the context of community detection, it has been argued that  
       nonbacktracking measures will completely ``ignore'' the presence 
        of such trees, with undesirable consequences  \cite{New13}.
        Hence, in this centrality measurement setting it is of interest to
         see whether a similar effect arises. 
         Intuitively, if we count only nonbacktracking walks, then 
          the connections 1-2, 1-3, 1-4 and 1-5 possessed by node 1 should be 
           less valuable than the connections enjoyed by the 
             other nodes in the graph.
           Hence, the small $\theta$ regime   should not be favourable
           for node 1.

         \begin{figure}
  \begin{center}
\includegraphics[width=0.25\textwidth]{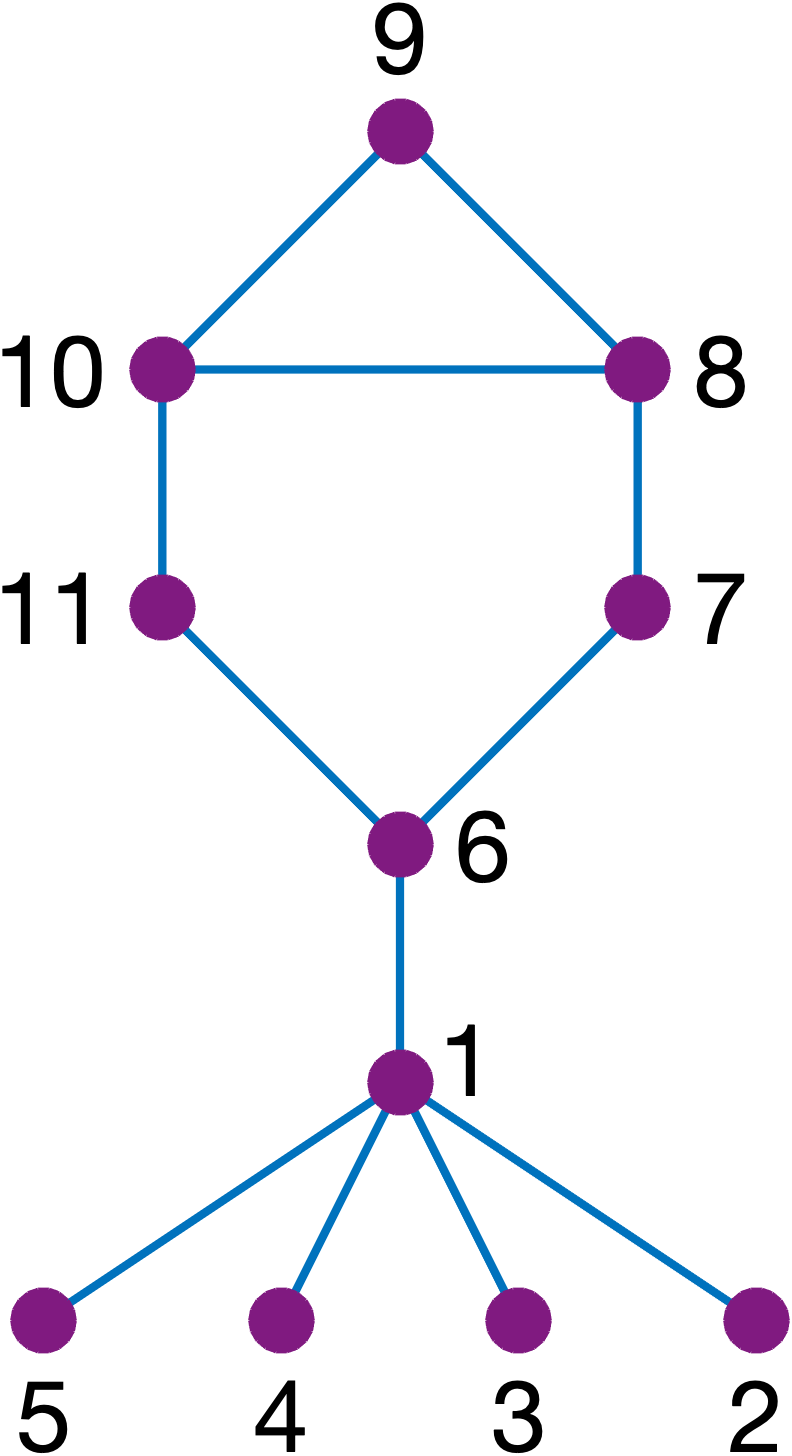}
\end{center}
\caption{The squid graph. 
}
\label{fig:squid}
\end{figure}

Figure~\ref{fig:res_squid} shows the 
 BTDW Katz centrality measure
 in (\ref{eq:katzbtdw}), normalized to have 
$\| \mathbf{x} \|_1 = 1$,  
  for nodes 1 (solid), 6 (dashed) and 8 (dotted).
Here, we used a fixed value of $\alpha = 0.99/\rho(A)$, which corresponds to 
$\alpha \approx 0.39$, and we show the centrality values
as $\theta$ ranges between $0$ and $1$.
 (Of course, by symmetry node 10 will always have the same centrality value as
node 8.)
The plot reveals a crossover effect, where 
a sufficiently small value of $\theta$, that is, sufficiently stringent 
downweighting of 
backtracking walks, causes node 1 to be ranked below nodes 6 and 8.
It is also of note that node 8, which is able to take part in 
a three-cycle, and is therefore involved in a relatively large
number of short nonbacktracking walks, is rated more highly 
than node 6 for small $\theta$. 
However, as $\theta$ increases,  
and hence the downweighting  of backtracking becomes less
severe,
node 8 is overtaken by 
 node 6, which arguably occupies a more 
central position in the graph.

 \begin{figure}
  \begin{center}
\includegraphics[width=0.5\textwidth]{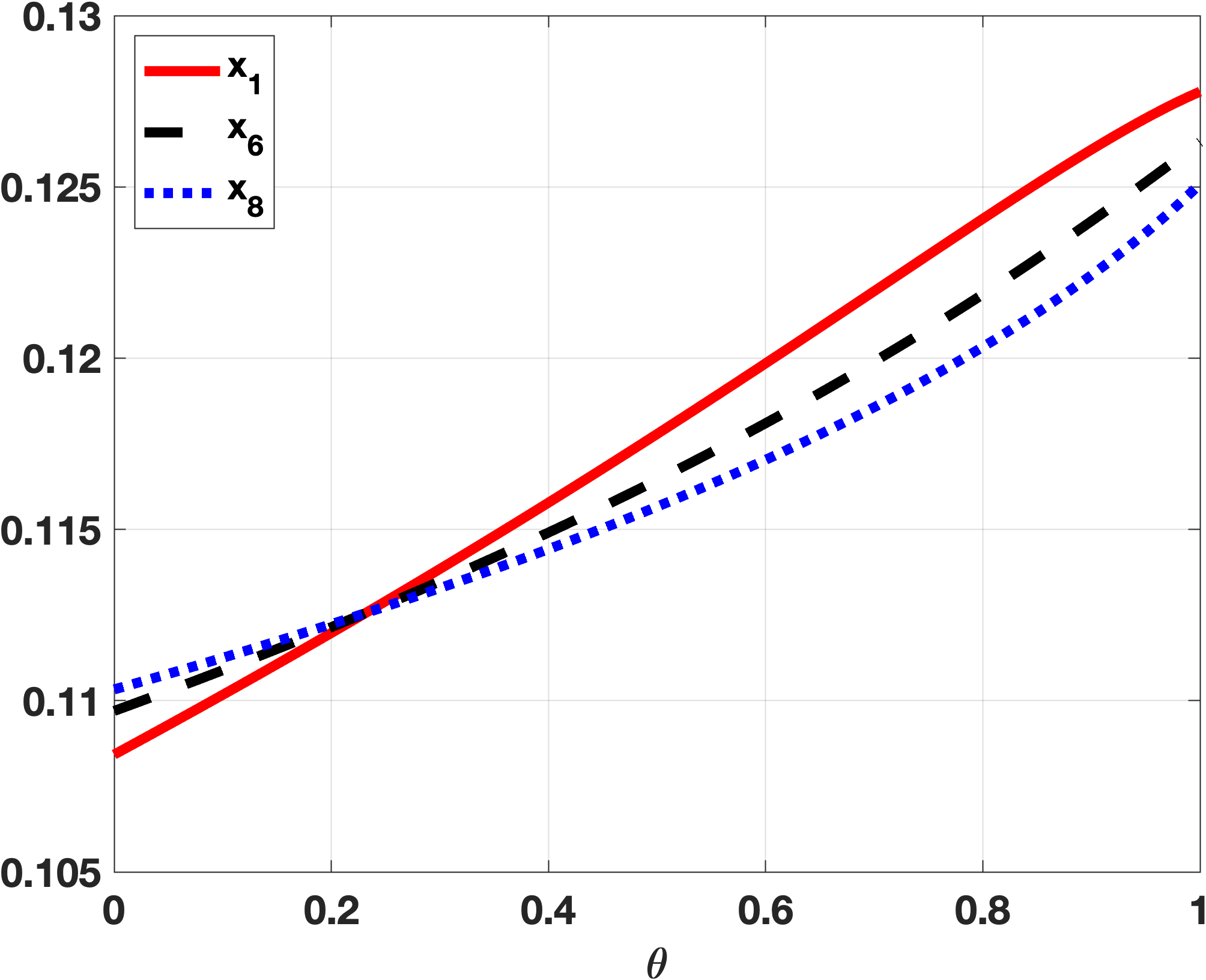}
\end{center}
\caption{Normalized
BTDW Katz centrality for nodes 
 nodes 1 (solid), 6 (dashed) and 8 (dotted), 
 as a function of $\theta$, 
 corresponding to the squid graph 
 in Figure~\ref{fig:res_squid}.
Here $\alpha$ is fixed at $0.99/\rho(A)$.
}
\label{fig:res_squid}
\end{figure}

  \begin{remark}\label{rem:squid}
  We mention in passing that the squid graph in Figure~\ref{fig:squid}  has an unusual 
  property: it possesses nodes that share    
  exactly the same eigenvector centrality (so that they could be described 
  as  spectrally iso-central) 
  whilst being topologically distinct.
   More precisely, let  $\lambda$ and  $\mathbf{v}$ denote the
    Perron-Frobenius eigenvalue and eigenvector associated with the adjacency matrix
    of this graph, respectively.
    Then straightforward algebra confirms that 
       $\lambda = (1+\sqrt{17})/2$
        and,
       after normalizing so that $v_1 = 1$,
       \[
v_1 = v_6 = v_8 = v_{10} = 1
\]
are the largest components,
 \[
v_7 = v_9 = v_{11} = \frac{\lambda-1}{2} 
\]  
are the next-largest components,
and 
\[
v_2 = v_3 = v_4 = v_5 = \frac{\lambda-1}{4} 
\]
are the 
smallest components.
    \end{remark}
     
 Next, we study the star graph 
  with $n = m+1$ nodes, $S_{1,m}$.
  We label the nodes so that node $1$ is the hub, that is, 
   the only node of degree $m$.
  The remaining nodes, which have degree one, are connected only to the hub.
  All edges are undirected.
   We note that 
  the star graph is a widely used test case for centrality measures \cite{BV14,Free78}, and we also point out that 
    the eigenvector version of full nonbacktracking centrality, 
    \cite{MZN14}, 
    breaks down on a star graph \cite{GHN18}. 
    The following theorem characterizes 
     the matrices 
     $q_k(A)$ that arise.
     
 \begin{theorem}\label{thm:star}
Let $A$ be adjacency matrix of the star graph with $n = m+1$ nodes, $S_{1,m}$. 
Then $q_0(A) = I$, $q_1(A) = A$, $q_2(A) = A^2 -(1-\theta)D$ and 
generally 
\begin{itemize}
\item[(i)] for all $k = 0,1,2,\ldots$
\[
q_{2k+1}(A) = \eta^k A
\]
where $\eta = \theta(\theta + m -1)$; and 
\item[(ii)] for all $k=2,3,\ldots$ and for $\theta\neq 0$ 
\[
q_{2k}(A) = \eta^{k-1}A^2\left(I - \left(\frac{1-\theta}{\eta}B\right)^k\right)\left(I - \frac{1-\theta}{\eta}B\right)^{-1} - (1-\theta)^kDB^{k-1},
\]
where $B = (1-\theta)I - D$.
\end{itemize}
\end{theorem}
     
\begin{proof}
See Appendix~\ref{sec:star}.
\end{proof}

We then have the following expression for 
BTDW Katz centrality.

\begin{corollary} \label{cor:Katzstar}  
Consider the star graph with
$n = m+1$ nodes, $S_{1,m}$, 
and let 
$\eta = \theta(\theta + m -1)$.
On this graph, 
the  BTDW Katz centrality measure (\ref{eq:katzibtdw}) 
exists
for
$\alpha^2 \eta < 1$ and 
has the form 
\[\left(\sum_{k=0}^\infty \alpha^kq_k(A)\bone\right)_i = 
1+ \frac{\alpha}{1-\alpha^2\eta}
\left.\begin{cases}
\left(
m(1+\alpha\theta)\right)  & \text{ for }i=1 \\
\left( 1+\alpha(\theta+m-1) \right) & \text{ for } i=2,\ldots,m+1.
\end{cases}\right.
\]
  \end{corollary}

\begin{proof}
See Appendix~\ref{sec:star}.
\end{proof}

It can in fact be formally proved that the radius of convergence of the generating function of BTDW Katz centrality is indeed $\eta^{-1/2}$; this will be done in a manuscript in preparation.  It follows from Corollary~\ref{cor:Katzstar} that for large $n$ 
\[
\frac{x_1}{x_2}
\approx 
\frac{1}{\alpha}+ \theta.
\]
So, for a fixed $\alpha$, the extent to which the hub node is prioritized over a leaf node decreases as we penalize backtracking. 
Also of interest is the regime of large $n$ and fixed $\theta$, where 
the range of allowable $\alpha$ values in 
Corollary~\ref{cor:Katzstar} 
is $(0,1/\sqrt{\theta^2+\theta m - \theta})$, i.e., approximately
$(0,1/\sqrt{\theta n})$; this shows how the
$\alpha$ range may increase substantially  as we downweight the backtracking.

At the other extreme, we may consider an undirected  $d$-regular graph where $d$ is large. (Here, 
$\mathbf{x} \propto \mathbf{1}$
solves  \eqref{eq:katzbtdw} for all values of $\theta  \in [0,1]$,
so all nodes are ranked equally, but it is informative to 
study the singularity of the coefficient matrix.)
For a 
$d$-regular graph any walk of length $k$ may be extended to a walk of length $k+1$
 using $d-1$ nonbacktracking edges and only one backtracking edge.
 So we would expect the allowable range of $\alpha$ values to increase
 much less dramatically as we decrease $\theta$. Indeed, 
since $\rho(A) = d$,
as $\alpha$ increases
from zero
in 
 \eqref{eq:katzbtdw}
 it may be shown that 
 the system becomes
 singular when 
 $ \alpha = (d-\mu)^{-1}=(d-1+\theta)^{-1}$.
 Hence, in this case, the use of $\theta$ makes very little difference to  
  the range of allowable $\alpha$ values.

In future work, 
it would be of interest to study how 
the choice of 
$\theta$ impacts the 
allowable range of 
$\alpha$ values
for other types of graph and also for 
networks arising in practice.
  For the 
 transport network
 in section~\ref{sec:compex}, using the 
 spectral bound 
 in Theorem~\ref{thm:spec},
 we found computationally that
 the upper limit,
 $\alpha^\star$, varied approximately linearly between 
 $\alpha^\star = 0.264$ 
  at $\theta = 1$
  and  
   $\alpha^\star = 0.415$ 
  at $\theta = 0$.

  \section{Spectral Limit}\label{sec:spec}
In this section we briefly relate the Katz-style centrality measure to  
 an eigenvalue version.
 We begin by noting that 
      the recurrences (\ref{eq:qrec}) and 
      (\ref{eq:qrecb}) are closely connected with the block matrix
      $Z \in \RR^{3 n \times 3 n}$ of the form 
\begin{equation}
     Z = 
      \left[
       \begin{array}{ccc}
        0 & I & 0 \\
        0 & 0 & I \\
          -\mu^2 \left(A-S\right) & 
             \mu \left( \mu I - D \right)
             & 
             A
         \end{array}
         \right],
          \label{eq:Zdef}
         \end{equation}
         as made clear by the following theorem.
        
         \begin{theorem}\label{thm:rad}
          The power series defining the generating function $\Psi(A)$
           in (\ref{eq:psi}) is convergent when
           $
              \alpha < 1/\rho(Z)
            $.
\end{theorem}

\begin{proof}
Suppose $\alpha \rho(Z) < 1$. 
Let 
$v_k \in \RR^{3n \times n}$ be defined by 
\[
v_k
=
\alpha^k
\left[
 \begin{array}{r}
   q_{k-1}(A) \\
   q_k(A) \\
     q_{k+1}(A)
    \end{array}
    \right],
\]
for $k \geq 1$. Then we see from (\ref{eq:qrecb}) and (\ref{eq:Zdef}) that 
\[
 v_{k+1} = \alpha Z v_{k} = (\alpha Z)^k v_1 .
 \]
By Gelfand's formula \cite[Corollary 5.6.14]{HJ} it follows that, for any matrix norm $\| \cdot \|$ and for any $\epsilon > 0$, there exists $ k_0 \in \mathbb{N}$ such that if $k \geq k_0$ then $\| (\alpha Z)^k \| < \left(\rho(\alpha Z) + \epsilon\right)^k.$
Taking $\epsilon = (1-\rho(\alpha Z))/2$ and specializing to any submultiplicative matrix norm, we conclude that there exists $k_0$ such that
\[  k \geq k_0 \geq 1 \Rightarrow \alpha^k \|q_k(A)\| \leq \|v_{k+1}    \|  \leq \left( \frac{1+\rho(\alpha Z)}{2} \right)^k    \| v_1 \| .\]
The result follows because
\[ \left\| \sum_{k=0}^\infty \alpha^k q_k(A) \right\| \leq \sum_{k=0}^{k_0-1} \alpha^k \| q_k(A) \| + \| v_1 \| \sum_{k=k_0}^\infty \left(  \frac{1+\rho(\alpha Z)}{2} \right)^k  ,    \]
and since $\rho(\alpha Z) < 1$ the right hand side is convergent.
\end{proof}    
     
   \begin{remark}\label{rem:specZ}
   We note that 
 although 
   the bound $1/\rho(Z)$ 
in Theorem~\ref{thm:spec}    
   will generally be sharp, there exist 
   cases where this is not so.
   For example, in the star graph example of 
   Corollary~\ref{cor:Katzstar}, it may be shown that 
    the radius of convergence 
     is always 
     $ 1/ \sqrt{\eta}$, but this 
     quantity coincides with 
    $1/\rho(Z)$  if and only if 
   $\theta \ge 1/(m+1)$.
   This statement, along with further results that may be derived using
    matrix polynomial theory, will be proved in forthcoming work.  
   \end{remark}
   
   The next theorem characterizes the node ranking that arises 
    generically when the 
   radius of convergence is approached.

   \begin{theorem}\label{thm:spec}
   Suppose that the radius of convergence is $1/\rho(Z)$,
   and hence the bound in Theorem~\ref{thm:rad} is sharp.
    Also, 
   for fixed $\theta$,
  suppose that
  $ I - \rho(Z)^{-1} Z$ has rank $n-1$.    
  Then, as
 $ \alpha \to 1/\rho(Z)$  
 from below
the ranking produced by 
(\ref{eq:katzbtdw}) 
 converges to that given by the last n components of the unique
eigenvector of $Z$ associated with the eigenvalue $\rho(Z)$.
\end{theorem}
   \begin{proof}
   The proof of 
   \cite[Theorem~6.1]{AGHN17a}
  can be extended directly to this case. 
\end{proof}

   \begin{remark}
     Theorem~\ref{thm:spec} shows that the last $n$ components of the 
     dominant eigenvector of $Z$ in (\ref{eq:Zdef}) 
    is an appropriate backtrack-downweighted eigenvalue centrality measure.
    Indeed, for  $\theta = 0$ it reduces to the full backtracking 
     measure given in \cite{AGHN17a} for general graphs and in 
      \cite{MZN14} for undirected graphs.
      For general $\theta$, in the undirected case it reduces to the measure in  
       \cite[Theorem~3.8]{CFGGP19}. 
       \end{remark}

      \section{Exponential and Other Generating Functions}\label{sec:expo}
Katz centrality in (\ref{eq:katzi}) and (\ref{eq:katz}) is associated with the matrix 
resolvent function $(I-\alpha A)^{-1}$. 
Several authors have argued that other matrix functions, defined via different power series,  may also be useful; see, for example,
\cite{BK13,BK15,EHSiamRev}.
Hence, in this section,
given coefficients $\{c_k\}_{k=0}^{\infty}$,
           where $c_k$ is the downweighting factor associated with a walk of 
            length $k$, we are interested in characterizing 
         and computing the corresponding quantity
          $
             \sum_{k=0}^{\infty}  c_k q_k(A) 
             $,
              and the action of this matrix on $\mathbf{1}$.

          We define
          \begin{equation}
             f_0(y) = \sum_{k=0}^{\infty}  c_k y^k, 
             \label{eq:fzero}
             \end{equation}
             and, more generally, for any integer $s\ge 0$, 
             \begin{equation}
             f_s(y) = \sum_{k=0}^{\infty}  c_{s+k} y^k.
             \label{eq:fs}
             \end{equation}
              The following theorem shows how 
                $
             \sum_{k=0}^{\infty}  c_k q_k(A) 
             $ can be expressed in terms of $f_0(Z)$ and $f_2(Z)$. 
             Consequently
            the backtrack-downweighted version of any matrix-function
             based  centrality measure can be computed whenever the 
              underlying matrix function is computable.  
               We note that  
              the series defining $f_0(Z)$ converges whenever the series defining $f_2(Z)$ converges, and vice versa.

         \begin{theorem}\label{thm:genfun}
         When the series represented below converge, we have
         \begin{equation}
          \sum_{k=0}^{\infty}  c_k q_k(A) = 
            \left[ 
              \begin{array}{ccc}
                0 & 0 &  I 
                \end{array}
                \right]
                f_0(Z)
                  \left[ 
              \begin{array}{c}
                I \\
                 A \\
                 A^2 - \mu D 
                \end{array}
                \right].
                \label{eq:fZgen}
                \end{equation}
                Moreover, this quantity may be computed as the $(3,3)$ block 
                of $f_0(Z) - \mu^2 f_2(Z)$.
                \end{theorem}
          
          \begin{proof}
          It follows directly from the recurrence (\ref{eq:qrecb}) and the definition 
           of $Z$ in 
           (\ref{eq:Zdef}) 
that, for all $k \ge 0$,           
\begin{equation}
Z^{k} 
  \left[
       \begin{array}{c}
         I  \\
          A \\
          A^2 - \mu D
         \end{array}
         \right]
   =         
    \left[
 \begin{array}{c}
   q_{k}(A) \\
   q_{k+1}(A) \\
   q_{k+2}(A)
    \end{array}
    \right].
    \label{eq:Zk}
\end{equation}
 The identity (\ref{eq:fZgen}) is then immediate.

          We next 
          define the block matrix $\widetilde Z \in \RR^{3 n \times 3n}$ by
\[
     \widetilde Z = 
     Z^2 - 
      \left[
       \begin{array}{ccl}
        0 & 0 & 0 \\
        0 & 0 & 0\\
        0 & 0 & 
          \mu^2 I 
         \end{array}
         \right], 
         \]        
          and note that 
      \begin{equation}
           \widetilde Z
              \left[ 
              \begin{array}{c}
                0 \\
                 0 \\
                 I
                \end{array}
                \right]
                 =
                 \left[
       \begin{array}{c}
         I  \\
          A \\
          A^2 - \mu D
         \end{array}
         \right].
             \label{eq:Zhat}
\end{equation}

Using (\ref{eq:Zk}) and (\ref{eq:Zhat}), 
we may write 
   \begin{eqnarray*}
           \sum_{k=0}^{\infty}  c_k q_k(A)            
            &=& 
    c_0 I
     + 
      c_1 A
      +
       \sum_{k=2}^{\infty} c_k q_k(A)\\
       &=&
         c_0 I
     + 
      c_1 A
      +
           \left[ 
              \begin{array}{ccc}
                0 & 0 &  I
                \end{array}
                \right]
                \sum_{k=2}^{\infty} 
                  c_k Z^{k-2} 
                     \widetilde Z
                       \left[ 
              \begin{array}{c}
                0\\
                 0 \\
                 I
                \end{array}
                \right]\\
                   &=&
         c_0 I
     + 
      c_1 A
      +
           \left[ 
              \begin{array}{ccc}
                0 & 0 &  I
                \end{array}
                \right]
                \sum_{k=2}^{\infty} 
                  c_k Z^{k} 
                       \left[ 
              \begin{array}{c}
                0\\
                 0 \\
                 I
                \end{array}
                \right]\\
                &&
   \mbox{} -      
     \left[ 
              \begin{array}{ccc}
                0 & 0 &  I
                \end{array}
                \right]
                \sum_{k=2}^{\infty} 
                  c_k Z^{k-2} 
                  \left[
       \begin{array}{ccl}
        0 & 0 & 0 \\
        0 & 0 & 0\\
        0 & 0 & 
          \mu^2  I 
         \end{array}
         \right]
                    \left[ 
              \begin{array}{c}
                0\\
                 0 \\
                 I
                \end{array}
                \right]   \\  
                &=&
                   \left[ 
              \begin{array}{ccc}
                0 & 0 &  I
                \end{array}
                \right]
                 f_0(Z)
                       \left[ 
              \begin{array}{c}
                0\\
                 0 \\
                 I
                \end{array}
                \right]
                -
                 \mu^2
                  \left[ 
              \begin{array}{ccc}
                0 & 0 &  I
                \end{array}
                \right]
                 f_2(Z)
                       \left[ 
              \begin{array}{c}
                0\\
                 0 \\
                 I
                \end{array}
                \right],
                \end{eqnarray*}
                which completes the proof. 
\end{proof}

\begin{remark}\label{rem:ode}
We note that in the exponential case, where $c_k = \alpha^k/(k!)$, 
we may recover (\ref{eq:fZgen}) through a more direct route.
Regarding the power series
$\sum_{k=0}^{\infty} \alpha^k q_{k}(A)/(k!)$ 
 as a function of the parameter $\alpha$, say $F(\alpha)$,
 we have
\[
 F(\alpha) =
 \sum_{k=0}^{\infty} 
                   \frac{\alpha^{k} q_{k}(A)} { k!},
\quad
F'(\alpha) =
 \sum_{k=0}^{\infty} 
                   \frac{\alpha^{k} q_{k+1}(A)} { k!},
                   \quad 
  F''(\alpha) =
 \sum_{k=0}^{\infty} 
                   \frac{\alpha^{k} q_{k+2}(A)} { k!}. 
                  \]
                  Then multiplying by $\alpha^k$ in
                  (\ref{eq:qrecb}) and summing from
                  $k = 0$ to $\infty$ gives the matrix-valued linear third order ODE
                  \[
                   F'''(\alpha) = A F''(\alpha) +
                      \mu \left( \mu I - D \right) F'(\alpha)
                      - \mu^2 (A - S) F(\alpha).
                      \]
                  This may be written in block first order form
          \[
          \left[
              \begin{array}{c}
                F(\alpha)\\
                F'(\alpha)\\
                F''(\alpha)
                \end{array}
                \right]'
                =
                Z
                \left[
              \begin{array}{c}
                F(\alpha)\\
                F'(\alpha)\\
                F''(\alpha)
                \end{array}
                \right]
           \]
           and hence, using 
           $F(0) = q_0(A) = I$, $F'(0) = q_1(A)= A$ and 
           $F''(0) = q_2(A) = A^2 - \mu D$, we have 
           \[
              \left[
              \begin{array}{c}
                F(\alpha)\\
                F'(\alpha)\\
                F''(\alpha)
                \end{array}
                \right]
                =
                \exp(\alpha Z)
                 \left[
              \begin{array}{c}
                I\\
                A\\
                 A^2 - \mu D
                \end{array}
                \right],
                \]
                which is consistent with (\ref{eq:fZgen}).
           \end{remark}

\section{Computational Example}\label{sec:compex}

In this section we illustrate the 
 BTDW Katz centrality measure (\ref{eq:katzbtdw}) 
 on a real transport network from 
 \cite{DSGA14}  with further data supplied in 
\cite{CDT20}.
 Here, $n = 271$ nodes represent stations in the London Underground system,
 with 312 (undirected) edges denoting rail links.
 We note that such a network has many ``hanging trees'' where  underground 
 lines head away from the city centre. Hence, we may expect a fully nonbacktracking centrality measure to penalize geographically outlying stations
 too severely. On the other hand, since there are some well-connected city centre stations that intersect many underground lines, we may expect 
 traditional eigenvector centrality to display localization, where most 
 of the centrality mass is placed on a small subset of the nodes.
 
 To be concrete, we will quantify localization in terms of the 
\emph{inverse participation ratio}, $S(n)$.
Here, for a family of unit-norm
vectors of increasing dimension, 
that is, 
 $\mathbf{v} \in \RR^n$,
 with 
 $\|  \mathbf{v} \|_2 = 1$,
letting 
 \[
 S(n) = \sum_{i=1}^{n} v_i^4,
 \]
 we say the sequence 
 is localized if $S(n) = O(1)$ as $n \to \infty$
 \cite{GDOM12}.
 In the tests below, where $n$ is fixed,
 we use the size of $S(n)$ as a measure of localization
 when comparing results, with a larger 
inverse participation ratio
 indicating greater localization. 

In Figure~\ref{fig:ipr} we show how the inverse participation ratio
for the 
BTDW Katz centrality measure (\ref{eq:katzbtdw}) 
varies as a function of $\theta$ and $\alpha$.
For this network, $\rho(A) \approx 3.78$, so standard Katz, 
that is $\theta = 0$, requires
$\alpha < 1/\rho(A) := \alpha^\star \approx  0.264$.
In the figure, we use $\alpha$ values of 
$0.05$,
 $0.06$,
\ldots,
$0.26$.
 We see from the figure that 
 the measure dramatically increases in localization
 in the regime where backtracking is not suppressed ($\theta \approx 1$) 
 and we are close to an eigenvector measure ($\alpha \approx \alpha^\star$).
 For Figure~\ref{fig:kendall} we made use of independent data
 from  \cite{CDT20} that records the annual
 passenger usage at each station. We took data for the most recent year, 2017.
 The idea now is to regard passenger usage (the total number exiting or entering a station) as an indication of importance, and to check whether 
 this correlates with the centrality measure, which is computed only from 
information about the network structure.
 We used  Kendall's tau coefficient to quantify
 the correlation between passenger usage and centrality.
 (Spearman's  rank correlation coefficient, which is also widely used
for assessing rankings, gave similar results.)
We see that the most localized regime
according to Figure~\ref{fig:ipr}
is also the regime with the best correlation coefficient.
However, by varying the backtrack-downweighting parameter,
$\theta$, it is possible to achieve a reduction in 
localization.
Varying the Katz parameter, $\alpha$, instead,
gives   
 a more rapid
loss of correlation.

 \begin{figure}
  \begin{center}
\includegraphics[width=0.7\textwidth]{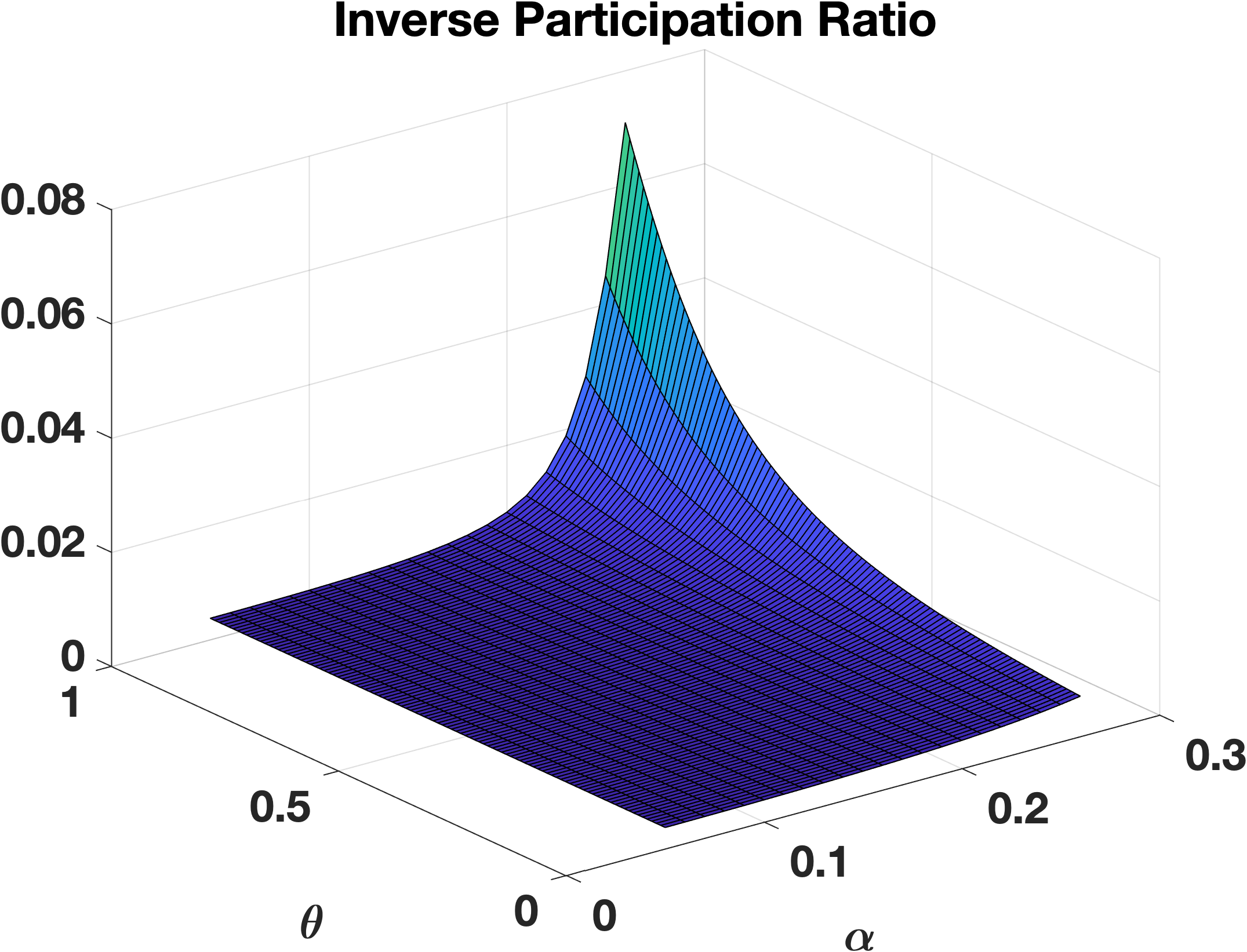}
\end{center}
\caption{
Inverse participation ratio for 
BTDW Katz centrality on the London Underground transport data.
}
\label{fig:ipr}
\end{figure}

\begin{figure}
  \begin{center}
\includegraphics[width=0.7\textwidth]{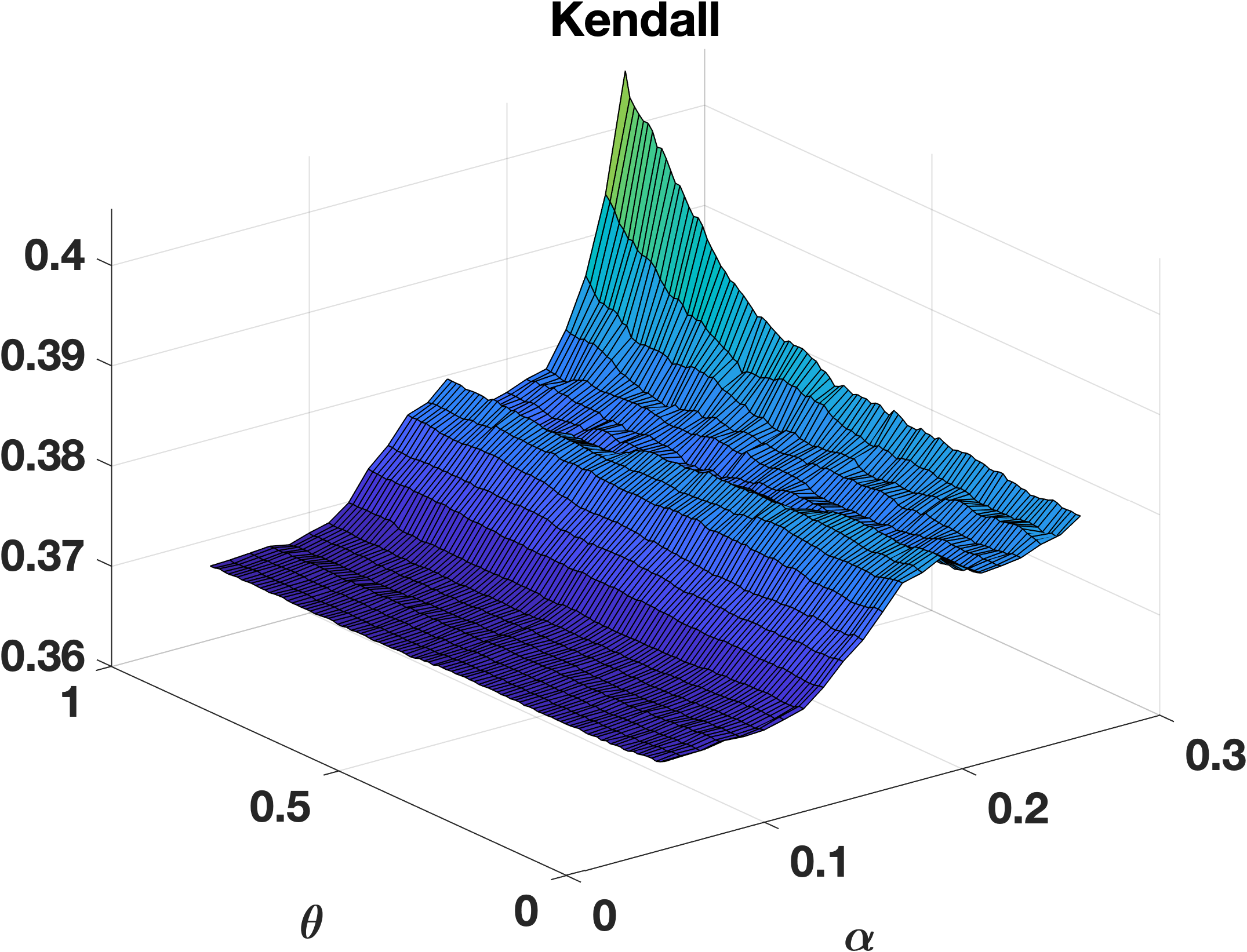}
\end{center}
\caption{
Kendall's tau coefficient for 
BTDW Katz centrality and annual passenger usage  on the 
London Underground transport data.
}
\label{fig:kendall}
\end{figure}

\section{Discussion}\label{sec:disc}
Our aim in this work was to define and study a framework that interpolates between 
traditional and nonbacktracking  
walk-counting 
 combinatorics.
 From a network science perspective, this has the potential to 
 combine 
 the benefits of both worlds---notably, 
 avoiding localization while accounting for
 tree-like structures. 
 Our results also 
 extend theoretical developments in 
 nonbacktracking walks and zeta functions on graphs from a range of 
 related fields
 \cite{AFH15,H07,HST,Sm07,ST96,WF09}.
 We developed a general
 four term recurrence in 
 Theorem~\ref{thm:qres} 
 and an expression for the 
 associated generating function in 
 Theorem~\ref{thm:genres}.
 In particular, we showed
in Corollary~\ref{cor:qKatz} that  
  the corresponding 
Katz-like centrality measure 
may be computed at the same cost as 
standard Katz.
By 
considering the relevant
limit,
Theorem~\ref{thm:spec} 
 then produced an eigenvector
centrality measure
that interpolates between 
the traditional and nonbacktracking extremes,

There is considerable scope for further work on backtrack-downweighted walks. 
In Remark~\ref{rem:specZ}
we quoted a counterintuitive result on the radius of convergence 
for the associated generating function
on a star graph. This raises questions such as how best 
to characterize the radius of convergence in general, and 
under what circumstances the lower bound in Theorem~\ref{thm:rad}
is sharp.
We are currently studying 
these issues 
with the tools of matrix polynomial theory.
From the perspective of algorithm design,
development of 
further insights 
and guidelines 
concerning the behaviour of backtrack-downweighted Katz 
in terms of the   
parameters 
$\alpha$ and $\theta$ 
is
 also of interest.

\bigskip
\bigskip


\bigskip
\bigskip

\noindent
\textbf{Acknowledgements}
The work of FA was supported by fellowship ECF-2018-453 from the Leverhulme Trust.
The work of DJH was supported by EPSRC Programme Grant EP/P020720/1.  The work of VN was supported by an Academy of Finland grant (Suomen Akatemian p\"a\"at\"os 331240).

\appendix
\section{Star Graph Results}\label{sec:star}

\noindent
\textbf{Proof of Theorem~\ref{thm:star}}

The expressions for $q_k(A)$ for $k=0,1,2$ are independent of the structure of the network; they follow immediately from Theorem~\ref{thm:qres}. 

The adjacency matrix has the form 
\[
A = \begin{bmatrix} 
0 & \bone^T  \\
\bone & \\
\end{bmatrix}.\]

\begin{itemize}
\item[(i)] Let $V = \{1,2,\ldots,m,m+1\}$ be the set of nodes in $S_{1,m}$. The graph is bipartite with node partitions $V_1 = \{1\}$ and $V_2 = \{2,\ldots,m+1\}$. It follows that all walks of odd length have to originate from a node in $V_i$ and terminate in a node in $V_j$, for $i,j=1,2$ and $i\neq j$. This immediately implies that, for all $k=0,1,\ldots$, the matrices $q_{2k+1}(A)$ have the same sparsity pattern as the matrix $A$: 
\[
q_{2k+1}(A) = \begin{bmatrix}
0 & \beta \bone^T \\
\beta\bone^T &
\end{bmatrix}
\]
for some $\beta>0$. 

We now proceed by induction. When $k=0$ we have that $q_{2k+1}(A) = A = \eta^k A$. 
Suppose now that the result holds for all $\ell\leq k$. We want to show that 
\[
q_{2k+3}(A) = \eta^{k+1}A = [\theta(\theta+m-1)]^{k+1}A.
\]
We proceed entrywise by considering the walks of length $2k+3$ from node $1$ to node $i\neq 1$; These are of two types:
\begin{itemize}
\item First type:
\[
      \underbrace{ 1 \,  \star \, \star \, \cdots  \, 1 \, i}_{(2k+1)\checkmark} \, 1 \, i,
\]
of which there are $\eta^k\theta\theta = \eta^k\theta^2$. The multiplication by $\theta^2$ is used to account for the two backtracking steps introduced when moving from node $i$ to node $i$ in two steps. 
\item Second type:
\[
      \underbrace{ 1 \,  \star \, \star \, \cdots  \, 1 \, j}_{(2k+1)\checkmark} \, 1 \, i,
\]
for $j\neq i$. Of these there are $\eta^k\theta(m-1)$. The multiplication by $(m-1)$ accounts for all the possible choices of $k$, while the multiplication by $\theta$ accounts for the added backtracking steps. 
\end{itemize}
Overall, by summing these two contributions, it follows that $(q_{2k+3}(A))_{1i}=\eta^k(\theta^2+\theta(m-1)) = \eta^{k+1}$, which concludes this part of the proof.
\item[(ii)] Let $k\geq 2$, then from Theorem~\ref{thm:qres} 
\begin{align*}
q_{2k}(A) &= q_{2k-1}(A)A + (1-\theta)q_{2(k-1)}(A)B \\
		&= \eta^{k-1}A^2 + (1-\theta)[q_{2k-3}(A)A+(1-\theta)q_{2(k-2)}(A)B]B \\
		&=\eta^{k-1}A^2 + (1-\theta)\eta^{k-2}A^2B + (1-\theta)^2q_{2(k-2)}(A)B^2 \\
		&=A^2\sum_{j=0}^{k-3}\eta^{k-1-j}((1-\theta)B)^j+(1-\theta)^{k-2}[q_3(A)A+(1-\theta)q_2(A)B]B^{k-2}\\
		&=A^2\sum_{j=0}^{k-1}\eta^{k-1-j}((1-\theta)B)^j - (1-\theta)^kDB^{k-1} \\
		&=\eta^{k-1}A^2\left(I - \left(\frac{1-\theta}{\eta}B\right)^k\right)\left(I - \frac{1-\theta}{\eta}B\right)^{-1} - (1-\theta)^kDB^{k-1},
\end{align*}
where we have used the fact that  
\[
B= (1-\theta)I-D = -\begin{bmatrix} 
\theta+m-1 &  \\
  & \theta I \\
\end{bmatrix}
\] 
and hence the matrix $I - \frac{1-\theta}{\eta}B$ is invertible.
\end{itemize}

\medskip

\noindent
\textbf{Proof of Corollary~\ref{cor:Katzstar}}

Corollary~\ref{cor:Katzstar} now follows 
after inserting the
expressions for $q_k(A)$ in Theorem~\ref{thm:star} 
into (\ref{eq:katzibtdw}) 
and simplifying, for all $\theta\in(0,1]$. When $\theta = 0$, i.e., in the fully nonbacktracking setting,
the result follows from a direct computation, since the length of the longest possible nonbacktracking walk in $S_{1,m}$ is two:
\[
x_1 = 1 + \alpha m,\quad \text{ and }\quad x_i = 1 + \alpha +\alpha^2 (m-1) 
\] 
for all $i=2,\ldots,m+1$. 

\bibliographystyle{siam}

\end{document}